\documentclass[letterpaper,twocolumn,10pt]{article}
\usepackage{usenix2019_v3}

\usepackage{amsmath,amssymb,amsfonts,amsthm}
\usepackage{centernot}
\usepackage{graphicx}
\usepackage{booktabs}
\usepackage{tabularx}
\usepackage{array}
\usepackage{float}

\makeatletter
\newcommand{\FloatBarrier}{%
  \ifx\@deferlist\@empty\else
    \clearpage
  \fi}
\makeatother

\usepackage{listings}
\usepackage{algorithm}
\usepackage{algpseudocode}
\usepackage{subcaption}
\usepackage{xspace}
\usepackage{xcolor}
\usepackage{tikz}
\usetikzlibrary{arrows.meta, backgrounds, positioning, calc, shapes.geometric, fit, shadows, matrix, patterns, decorations.pathreplacing}
\microtypesetup{spacing=false}
\AtBeginDocument{\DeclareMathAlphabet{\mathcal}{OMS}{cmsy}{m}{n}}

\newcolumntype{L}[1]{>{\raggedright\arraybackslash}p{#1}}
\newcolumntype{C}[1]{>{\centering\arraybackslash}p{#1}}
\newcolumntype{R}[1]{>{\raggedleft\arraybackslash}p{#1}}
\newcolumntype{Y}{>{\centering\arraybackslash}X}

\newcommand{\pdds}{\mbox{PDDS}\xspace}

\theoremstyle{definition}
\newtheorem{theorem}{Theorem}
\newtheorem{definition}{Definition}
\newtheorem{proposition}{Proposition}

\theoremstyle{definition}

\newtheorem{property}{Property}

\definecolor{slate}{RGB}{112,128,144}
\definecolor{emerald}{RGB}{16,163,127}
\definecolor{navy}{RGB}{24,49,83}
\definecolor{crimson}{RGB}{190,30,45}
\definecolor{amber}{RGB}{217,119,6}
\definecolor{cobalt}{RGB}{29,78,216}
\definecolor{teal}{RGB}{13,148,136}
\definecolor{purple}{RGB}{126,34,206}
\definecolor{softgray}{RGB}{247,249,252}
\definecolor{bordergray}{RGB}{218,224,233}
\definecolor{darkslate}{RGB}{60,72,88}
\definecolor{linegray}{RGB}{140,155,175}

\hypersetup{
  pdftitle={Cognitive Admission Control: Risk-Conditioned Assurance for Consequential Actions in Agentic Distributed Systems},
  pdfauthor={Jun He; Deying Yu},
  pdfsubject={Control-plane abstraction conditioning agent execution on epistemic assurance obligations in agentic distributed systems},
  pdfkeywords={cognitive admission control, consequential actions, agentic distributed systems, epistemic obligations, assurance, execution boundary, admission certificate, post-deterministic systems}
}

\begin{document}
\raggedbottom

\title{\bf Cognitive Admission Control:\\Risk-Conditioned Assurance for Consequential Actions\\in Agentic Distributed Systems}

\ifdefined\CACAnonymous
\author{Anonymous authors}
\hypersetup{pdfauthor={}}
\else
\author{
  {\rm Jun He}\\
  OpenKedge.io
  \and
  {\rm Deying Yu}\\
  OpenKedge.io
}
\fi

\maketitle

\begin{abstract}
In agentic distributed systems, an agent may be authorized to mutate external infrastructure while lacking evidence that the mutation is ready to execute. Cognitive Admission Control (CAC) makes this evidence requirement explicit. A policy maps a typed action and its modeled risk to assurance obligations specifying predicates, evidence classes, scope, freshness, and witness-set constraints. A deterministic evaluator distinguishes satisfied, violated, and unresolved obligations; unresolved conditions produce targeted evidence-acquisition requests. Successful admission produces a certificate binding the action, its witness manifest, and dispatch-time guards.

We formalize the admission calculus and the assumptions connecting it to mediated execution. The guarantees are policy-relative: physical safety additionally requires sound evidence, an adequate environment model, and preservation of relevant conditions through the effect. A TypeScript prototype is evaluated in 2,730 controlled local trials with independent effect observation and matched fault schedules. Across 390 CAC trials, 120 effects complete without modeled harm and no harmful effects occur. A live-policy baseline achieves the same completion count but admits the constructed correlated-witness failure. Mechanism ablations isolate guard, evidence-class, structural-cut, and remediation behavior. A further 9,000 measurements exercise the complete local dispatch path with persistent replay protection. These results establish tested implementation behaviors and local costs, not production failure rates or comparisons of language-model capability.
\end{abstract}

\section{Introduction}
\label{sec:introduction}

In agentic distributed systems~\cite{he2026pdds}, an autonomous agent authorized to promote a database replica may still lack evidence that the replica is current, that the old primary is fenced, or that the proposed target matches the approved scope. The access decision and the operational preconditions answer different questions. This distinction matters when a runtime accepts a generated tool call whose supporting observations may be incomplete, stale, or derived from a shared faulty source.

Existing systems already separate untrusted computation from trusted enforcement. Access-control engines evaluate policy predicates, proof-carrying code checks consumer-specified safety conditions, and agent transaction runtimes validate proposed work before commitment~\cite{saltzer1975protection,necula1997pcc,mnemosyne2026}. The problem addressed here is how to make \emph{missing evidence} an explicit part of this interface: which evidence must an action supply, what remains unresolved, and which parts of a successful evaluation must still hold when execution begins?

We propose \textbf{Cognitive Admission Control (CAC)}, a policy-governed admission layer for consequential actions in agentic distributed systems. CAC maps a typed proposal and a controller-visible risk profile to a set of assurance obligations. Each obligation specifies a predicate, acceptable evidence classes, scope, freshness, and any set-level requirements. The evaluator returns satisfied, violated, or unknown. Missing evidence produces a remediation contract; affirmative violation produces denial. A successful evaluation yields a certificate binding the proposal, the supporting witnesses, and dispatch-time guards.

Consider a failover policy requiring a receipt for replication state and a separate receipt confirming primary fencing. A fresh lag measurement cannot discharge the fencing obligation, regardless of how many times the agent repeats it or how much reasoning it performs. Once both obligations are satisfied, the certificate remains usable only within its scope and validity interval. If a bound state version changes, the gateway rejects dispatch. Preserving the same conditions until the database applies the effect requires target-side conditional execution or a transaction protocol.

This design makes a policy's assurance requirements inspectable and actionable. It does not make the policy correct by construction. An omitted dependency can produce an incomplete obligation set, and an authentic sensor can report a false value. We therefore separate policy-relative admission, mediated dispatch, and physical-world safety throughout the model.

\subsection{Contributions}
\label{sec:contributions}

The paper makes three contributions:
\begin{enumerate}
    \item A typed admission calculus connecting risk-conditioned obligation selection, ternary evidence discharge, and targeted evidence acquisition (Sections~\ref{sec:system-action-risk-evidence}--\ref{sec:cognitive-admission-control}).
    \item Explicit refinement conditions connecting a successful evaluation to a witness manifest, a bounded capability, and guarded dispatch. The formal results state the information and trust assumptions needed for these connections (Section~\ref{sec:cac-theory}).
    \item A TypeScript prototype with bounded witness search and persistent local replay protection, evaluated in 2,730 controlled trials and 9,000 complete-path timing samples. Independent effect observation, matched fault schedules, and mechanism ablations make the tested behaviors reproducible (Sections~\ref{sec:cacbench}--\ref{sec:performance-overhead}).
\end{enumerate}

The contribution is the admission interface and its composition rules, not a new access-control expressiveness result, a proof of general agent safety, or an empirical comparison of language models. The current evaluation establishes selected prototype behaviors; it does not establish a production safety rate.

\section{From Authority to Readiness}
\label{sec:authority-readiness}
\label{sec:auth-vs-admission}

Authorization and readiness can be evaluated by the same policy engine. Their distinction is informational: a check over identity and scope cannot determine an operational condition absent from its inputs. ABAC may incorporate dynamic state, and a sufficiently expressive policy engine can implement CAC's predicates. CAC specifies the evidence lifecycle around those predicates: acquisition, provenance, conflict handling, expiration, and binding to dispatch.

\subsection{Resources, Evidence, and Discharge}
\label{sec:resources-not-assurance}

Let $R$ bound tokens, context, retrieval calls, inspection calls, verification compute, elapsed time, and cost:
\begin{equation}
R=\langle T,C_{\mathrm{ctx}},M,A_{\mathrm{tool}},V_{\mathrm{cpu}},L,\$\rangle.
\label{eq:resource-vector}
\end{equation}
A process consumes resources to produce evidence $E$; the evaluator checks $E$ against required obligations $\Omega$:
\begin{equation}
R \xrightarrow{\text{acquisition}} E
  \xrightarrow{\text{discharge under }\Pi} \text{admission verdict}.
\label{eq:two-stage-transformation}
\end{equation}
More compute can improve reasoning~\cite{snell2024scaling,wei2022chain}. Its allocation alone does not certify that a particular observation was made or remains current. Likewise, a database role observation does not establish replication completeness: distinct predicates require evidence with the appropriate semantics.

\subsection{Independent Enforcement}
\label{sec:confidence-vs-evidence}
\label{sec:agent-vs-controlplane}

The proposer may generate proofs, collect telemetry, or suggest verification steps. It cannot set the evaluator's required predicates or unilaterally mark them satisfied. Model confidence is usable only if the governing policy explicitly defines and justifies its evidentiary role; confidence is not a general substitute for a fresh external observation.

The admission verdict and canonical witness selection are deterministic for fixed proposal, modeled state, evidence, and evaluation context. Certificate issuance additionally uses an issuance time and a fresh nonce, so repeated evaluations need not produce byte-identical certificates. The trust boundary separates the evidence-producing process from the code that interprets evidence and authorizes dispatch.

\subsection{Consequential and Lightweight Paths}
\label{sec:non-consequential-routing}

Policy classifies operations by their modeled consequences. Operations below its threshold may use ordinary authorization and lightweight mediation. Consequential actions require the obligation-based admission path. Read-only operations are not automatically low-risk: disclosure, resource contention, or diagnostic side effects may make them consequential. Classification belongs to the trusted control plane; a classification miss is outside the conditional guarantee for correctly classified actions.

\section{System, Risk, and Evidence Model}
\label{sec:system-action-risk-evidence}
\label{sec:system-model-formal}

\subsection{State and Trust Boundary}

An agent proposes operations to an execution gateway $\mathcal K$. The agent may produce incorrect plans or evidence, but cannot modify the trusted evaluator, signing keys, policy, or gateway. Complete mediation requires that target credentials and network access remain confined to $\mathcal K$. Identity, observation channels, policy configuration, and the gateway's state sources are explicit trust dependencies, not properties established by the agent.

Let $\mathcal S^*$ be the external state space and $\mathcal S$ the space of controller-visible snapshots. A valid snapshot $s\in\mathcal S$ represents a nonempty set $\gamma_{\mathcal A}(s)\subseteq\mathcal S^*$ of possible external states under environment model $\mathcal A$. The live state $s^*$ need not be uniquely determined by $s$. Model fidelity requires $s^*\in\gamma_{\mathcal A}(s)$ and coverage of relevant external transitions. An inconsistent snapshot with empty concretization is rejected, not treated as vacuously safe. An agent's internal belief $\hat s$ is a separate object and supplies no privileged access to $s^*$.

\subsection{Intent and Action}
\label{sec:intent-proposal}

\begin{definition}[Intent and proposal]
An intent $i=\langle\text{goal},\text{scope},\text{constraints}\rangle$ specifies desired external states, permitted targets, and operational constraints. A proposal is
\begin{equation}
\begin{aligned}
q=\langle&i,a,\text{params},\text{scope},\\
&\text{principal},s_v,\text{constraints}\rangle,
\end{aligned}
\label{eq:proposal-def}
\end{equation}
where $a$ identifies a tool operation and $s_v$ is the state version used to formulate it. The evaluator checks a normalized, schema-valid proposal. Its canonical representation includes every field affecting interpretation or execution.
\end{definition}

\begin{definition}[Modeled consequence and safety]
Let $\operatorname{Post}_{\mathcal A}(q,s)$ contain all possible external post-states attributable to $q$ from states in $\gamma_{\mathcal A}(s)$ under the declared transition model. The proposal is consequential when its worst-case modeled loss crosses a policy threshold:
\begin{equation}
\sup_{\substack{x\in\gamma_{\mathcal A}(s)\\
y\in\operatorname{Post}_{\mathcal A}(q,x)}}
\operatorname{Loss}_{\Pi}(x,y)\ge\tau_{\Pi}.
\label{eq:consequential-action}
\end{equation}
Here $\operatorname{Post}_{\mathcal A}(q,x)$ denotes transitions from a particular external state $x$; their union defines $\operatorname{Post}_{\mathcal A}(q,s)$. Modeled safety is
\begin{equation}
\operatorname{Safe}_{\Pi}(q,s)\iff
\forall y\in\operatorname{Post}_{\mathcal A}(q,s):
\operatorname{Invariant}_{\Pi}(y).
\label{eq:safe-predicate-def}
\end{equation}
Loss thresholds classify actions; they do not themselves prove safety. Unknown or unsupported transition models require conservative classification or escalation.
\end{definition}

\subsection{Risk and Policy}
\label{sec:risk-model}

The controller computes a policy-relative risk vector
\begin{equation}
\rho(q,s)=\langle C,B_{\mathrm{blast}},I,U,D,P\rangle\in\mathcal R,
\label{eq:risk-vector-def}
\end{equation}
representing consequence severity, blast radius, irreversibility, observational uncertainty, dependency exposure, and adverse plausibility. Each coordinate has a declared preorder. The product preorder is $\rho_1\preceq\rho_2$ iff every coordinate of $\rho_1$ is below the corresponding coordinate of $\rho_2$. Profiles can be incomparable. These are policy categories, not calibrated probabilities unless an estimator is separately validated.

The policy snapshot $\Pi$ maps the proposal, snapshot, and risk to obligations. The mapping is trusted configuration. Determinism makes its output reproducible, not necessarily conservative or adequate.

\subsection{Evaluation Context and Receipts}
\label{sec:evidence-model}

\begin{definition}[Evaluation context]
\begin{equation}
\chi=\langle\Pi_{\mathrm{eval}},t_{\mathrm{eval}},B_{\mathrm{op}},
\mathcal D_{\mathrm{eval}},\mathcal A_{\mathrm{eval}}\rangle
\label{eq:eval-context}
\end{equation}
pins policy and source registrations, evaluation time, operational budget, dependency topology, and authorization state. Declarative admissibility is independent of $B_{\mathrm{op}}$; budget exhaustion may terminate the operational procedure without changing whether the evidence logically satisfies policy.
\end{definition}

All receipt timestamps must be comparable in a declared clock domain. We use a normalized time scale with no future-dated receipts. A distributed deployment must provide authenticated time normalization and conservative skew bounds; unrelated local monotonic-clock readings cannot be subtracted directly.

\begin{definition}[Evidence receipt]
\begin{equation}
\begin{aligned}
e=\langle&\text{id},\text{claim},\text{class},\text{source},\text{provenance},\\
&\text{scope},t_{\mathrm{obs}},\nu_s,\text{integrity},\text{dependencies}\rangle.
\end{aligned}
\label{eq:evidence-object}
\end{equation}
The receipt binds a typed claim to its source, observation time, scope, and state version. Receipt identifiers are unique and content-bound; duplicate identifiers with different contents are rejected. Authenticity establishes origin and integrity under the registered trust roots. It does not establish the truth of the claim.
\end{definition}

\subsection{Obligations and Witness Eligibility}
\label{sec:obligations-efd}

\begin{definition}[Assurance obligation]
\label{def:obligation}
An obligation $\omega$ declares a predicate, verification kind, scope, accepted evidence classes $\mathcal E_{\mathrm{class}}$, freshness horizon $\Delta t_{\max}$, set constraints, optional EFD requirement, guard templates, and enforcement mode (required or advisory). Typical kinds are observation, verification, simulation, quorum, and dual control. Its identifier and policy epoch bind these fields.
\label{eq:typed-obligation}
\end{definition}

A receipt is locally eligible exactly when it is authentic, belongs to an accepted class, covers the required scope, has a compatible state version, and satisfies
\begin{equation}
0\le t_{\mathrm{eval}}-e.t_{\mathrm{obs}}\le\omega.\Delta t_{\max}.
\label{eq:eligibility-def}
\end{equation}
We denote this conjunction by $\operatorname{LocallyEligible}_{\Pi}(e,\omega,s,t_{\mathrm{eval}})$ and its candidate pool by $\mathcal E_\omega$. Rejection diagnostics may overlap: an untrusted receipt can also be stale.

For $W\subseteq\mathcal E_\omega$, $\operatorname{WitnessValid}_{\Pi}(W,\omega,s)$ checks the declared set constraints, such as distinct-source cardinality, principal separation, or reconciliation. Counting receipts does not count independent sources. Separate entailment relations determine whether $W$ supports the predicate or its negation. Both polarities must satisfy their declared witness-validity requirements; a single negative receipt is not automatically a valid quorum counter-witness. A policy may separately declare that one trusted veto is sufficient.

\subsection{Structural Epistemic Fault Domains}

Following EFD~\cite{he2026efd}, let $\mathcal B_E$ be a finite basis of modeled root faults, $\mathcal D(v)\subseteq\mathcal B_E$ the exposure set of verifier $v$, and $\Gamma$ a nonempty family of nonempty decisive coalitions of a panel $Q$. The structural cut is
\begin{equation}
\begin{aligned}
\kappa_E(Q,\Gamma)=\min\{|F|:\;&F\subseteq\mathcal B_E,\ \exists C\in\Gamma,\\
&\forall v\in C,\ F\cap\mathcal D(v)\ne\emptyset\}.
\end{aligned}
\label{eq:efd-cut-requirement}
\end{equation}
The minimum of an empty feasible set is $+\infty$. In a deployment, an unknown exposure profile must be rejected rather than interpreted as a fault-free verifier. The shorthand $\kappa_E(W)$ uses the distinct sources in $W$ and the obligation's declared coalition rule, proposition, and temporal context.

A cut of at least $k$ excludes decisive corruption by fewer than $k$ modeled roots only if every erroneous approval is covered by an active root, the exposure map is conservative, and $\Gamma$ matches the actual authorization rule. This is a structural guarantee, not a statement of statistical independence.

\section{Cognitive Admission Control}
\label{sec:cognitive-admission-control}

\subsection{Required Obligations}
\label{sec:obligation-derivation}

The policy selects obligations using the typed proposal and modeled state:
\begin{equation}
\Omega_{\Pi}(q,s)=F_{\Pi}(\rho(q,s),q,s).
\label{eq:obligation-derivation}
\end{equation}
Required obligations constrain admission; advisory obligations produce diagnostics. Policies may emit dual-control requirements when risk exceeds autonomous authority. An approval receipt cannot override a separate violated obligation.

\subsection{Discharge Semantics}
\label{sec:ternary-discharge}

For a required obligation $\omega$, let $\mathbb W_\omega^+$ and $\mathbb W_\omega^-$ be the finite sets of eligible, valid witness sets supporting its predicate and its negation, respectively. Entailment and validity are total policy-defined procedures over a restricted predicate language. They do not invoke an unrestricted theorem prover with an assumed termination guarantee.

A policy supplies a strict partial preference order on valid witnesses. Let $M_\omega$ be the maximal elements of $\mathbb W_\omega^+\cup\mathbb W_\omega^-$. Write $P_\omega$ and $N_\omega$ for the presence of positive and negative witnesses in $M_\omega$, respectively. The discharge result $D_\omega=\operatorname{Discharge}(\omega,E,s,\chi)$ is:
\begin{equation}
\begin{aligned}
D_\omega=
\begin{cases}
\textsc{Satisfied}(W) & P_\omega\land\neg N_\omega,\\
\textsc{Violated}(W) & N_\omega\land\neg P_\omega,\\
\textsc{Unknown}(\text{Conflict}) & P_\omega\land N_\omega,\\
\textsc{Unknown}(c) & \neg P_\omega\land\neg N_\omega.
\end{cases}
\end{aligned}
\label{eq:ternary-discharge}
\end{equation}
A witness cannot support both polarities under a well-formed entailment relation. Equal-priority or incomparable opposing witnesses produce conflict. Only after the polarity is established does the evaluator select a canonical witness from $M_\omega$.

When no valid witness exists, $c$ identifies a failed structural constraint, stale or future-dated evidence, untrusted provenance, incompatible version, or missing usable evidence. A fixed diagnostic precedence makes $c$ reproducible; it does not alter the semantic verdict. The policy must distinguish a missing observation from evidence that its predicate is false.

\paragraph{Canonical selection and finite search.}
Among maximal witnesses of the selected polarity, the specification ranks by policy priority, best worst-case freshness, structural cut when required, smallest sufficient cardinality, and a canonical encoding tie-breaker. This defines a unique result on finite inputs. It does not imply efficient exhaustive subset search: a pool of $n$ receipts has $2^n$ subsets. Implementations must restrict the witness grammar or bound search. A truncated search returns unknown or abort, never a permit inferred from incomplete conflict inspection. The prototype enumerates both polarities up to twelve unique eligible receipts per obligation and otherwise returns unknown. It implements a policy-defined priority order over exact source identifiers; witnesses with equal priority remain incomparable across polarities.

\subsection{Declarative Admissibility}
\label{sec:declarative-admissibility}

\begin{definition}[CAC admissibility]
$\operatorname{Admissible}_{\Pi}(q,s,E,\chi)$ holds when:
\begin{enumerate}
\item $q$ is schema-valid, bound to an injective canonical encoding, authorized under $\chi.\mathcal A_{\mathrm{eval}}$, and statically compliant with $\Pi=\chi.\Pi_{\mathrm{eval}}$;
\item the policy resolver enumerates every obligation triggered by the fixed inputs;
\item every required obligation returns $\textsc{Satisfied}(W_\omega)$ under the discharge rule;
\item the derived guard set $G_q$ and permitted invocation envelope $\mathcal E_q$ satisfy the coverage conditions below.
\end{enumerate}
Admissibility is independent of operational budget. It makes no assertion about obligations absent from the policy.
\end{definition}

\paragraph{Witness manifest.}
The manifest preserves the relation between each required obligation and its canonical witness:
\begin{equation}
\mathcal W_q=\{(\omega,W_\omega):
\omega\in\operatorname{Required}(\Omega_\Pi(q,s))\}.
\label{eq:admission-witness-def}
\end{equation}
Local eligibility supplies freshness and version compatibility. Set validity supplies the declared quorum and separation conditions. Advisory results are recorded separately.

\paragraph{Envelope coverage.}
An exact-action certificate allows only $q$. An envelope certificate may allow a variant $q'$, but parameter attenuation alone is insufficient: a different target can require different evidence. Define $\operatorname{ManifestCovers}_{\Pi}(\mathcal W_q,q',s,\chi)$ to require that every obligation for $q'$ has a compatible manifest witness which remains eligible, valid, and affirmatively entails that obligation under the same conflict semantics. Envelope soundness requires, for every allowed consequential $q'$,
\begin{equation}
\begin{aligned}
&\operatorname{ManifestCovers}_{\Pi}(\mathcal W_q,q',s,\chi)\\
&\quad\land\operatorname{RequiredGuards}_{\Pi}(q',s,\chi)\subseteq G_q.
\end{aligned}
\label{eq:envelope-sound}
\end{equation}
The controller restricts envelope mode to approved templates with independently justified semantic coverage. The template also fixes the subject and action interpretation and establishes static compliance for each allowed variant; the gateway separately rechecks live authorization. The controller checks template identity, policy epoch, parameter domain, coverage requirements, and guards online. Exact-action mode reduces coverage to $q$ itself, but still requires guard completeness. General envelope verification is not claimed to be decidable; signed template approval is a trust assumption, not a substitute for its semantic justification.

\subsection{Operational Decisions and Remediation}
\label{sec:operational-decision}
\label{sec:epistemic-loop-formal}

The controller returns one of five outcomes:
\begin{itemize}
\item $\textsc{Permit}(\mathcal C_q)$ when admissibility and certificate validity checks succeed;
\item $\textsc{Deny}$ for authorization, static-compliance, or decisive required-obligation violation;
\item $\textsc{Defer}$ with unresolved obligations and typed acquisition requests;
\item $\textsc{Escalate}$ when unresolved approval obligations need an authorized reviewer, retaining all other unresolved requirements;
\item $\textsc{Abort}$ when evaluation cannot finish within its operational bounds.
\end{itemize}
Exhaustion may prevent a permit even when the logical admission conditions are satisfiable. A fresh evaluation context is captured on each remediation turn; evidence obtained after one evaluation is reconsidered in the next. Finite budgets bound the number of attempts, but useful completion requires available evidence and a sufficiently stable environment.

A consequential diagnostic action must itself undergo admission or use a separately approved bounded diagnostic capability. Recursive evidence acquisition cannot bypass the execution boundary. Preventing cyclic dependencies requires an additional well-founded acquisition discipline; it is not implied by the discharge rules.

\subsection{Certificates and Expiration}
\label{sec:certificate-binding}

A certificate binds the subject, authority context, mode, proposal digest, envelope, policy epoch, witness digest, guards, issuance and expiration times, nonce, and signature. Digests use domain-separated canonical encodings:
\begin{equation}
\begin{aligned}
\operatorname{ProposalDigest}(q)&=\mathcal H(d_{\mathrm{prop}}\Vert\operatorname{Canon}(q)),\\
\operatorname{WitnessDigest}(\mathcal W_q)&=\mathcal H(d_{\mathrm{wit}}\Vert\operatorname{Canon}(\mathcal W_q)).
\end{aligned}
\label{eq:prop-digest}
\end{equation}
The distinct fixed tags $d_{\mathrm{prop}}$ and $d_{\mathrm{wit}}$ prevent cross-type interpretation. Canonicalization rejects ambiguous encodings and includes all security-relevant fields.

The earliest expiry of a supporting receipt bounds certificate lifetime:
\begin{equation}
t_{\mathrm{wit}}=\min_{(\omega,W)\in\mathcal W_q,\ e\in W}
(e.t_{\mathrm{obs}}+\omega.\Delta t_{\max}),
\label{eq:witness-expiry}
\end{equation}
with $+\infty$ for an empty manifest. The mint enforces
\begin{equation}
t_{\mathrm{issued}}\le t_{\mathrm{exp}}\le
\min(t_{\mathrm{wit}},t_{\mathrm{policy}},t_{\mathrm{auth}},t_{\mathrm{cap}}).
\label{eq:cert-expiry-bound}
\end{equation}
If the bound precedes issuance, admission must be refreshed. A distributed clock implementation must subtract any uncertainty needed to make this expiry conservative.

\subsection{Gateway Validation and Dispatch}

The gateway authenticates the presenter as the certified subject and validates the certificate without consuming it:
\begin{equation}
\begin{aligned}
&\operatorname{AdmissionValid}(\mathcal C,q',\chi_{\mathrm{live}})\iff\\
&\quad\operatorname{VerifySig}(\mathcal C)\\
&\quad\land\operatorname{PresenterAuthorized}(\mathcal C)\\
&\quad\land t_{\mathrm{issued}}\le t_{\mathrm{now}}\le t_{\mathrm{exp}}\\
&\quad\land\operatorname{NonceUnused}(\mathcal C)\\
&\quad\land\operatorname{EnvelopeValid}(\mathcal C,q')\\
&\quad\land\operatorname{GuardsHold}(\mathcal C.G,s_{\mathrm{live}})\\
&\quad\land\operatorname{PolicyEpochCompatible}(\mathcal C)\\
&\quad\land\operatorname{AuthorizationStillValid}(\mathcal C,q')\\
&\quad\land\operatorname{NotRevoked}(\mathcal C).
\end{aligned}
\label{eq:gateway-assertion}
\end{equation}
The live context fixes the authenticated requester, current time, policy, authorization, and observed guard state for this check. A comparison of untrusted identity strings is not authentication. Unknown guards fail closed.

After successful validation, the gateway performs an atomic unused-to-consumed nonce transition. Only the winner may dispatch:
\begin{equation}
\begin{aligned}
&\operatorname{AuthorizedDispatch}_{\mathcal K}\iff\\
&\quad\operatorname{AdmissionValid}(\mathcal C,q',\chi_{\mathrm{live}})\\
&\quad\land\operatorname{CAS}(\mathcal C.\text{nonce},
\text{Unused},\text{Consumed}).
\end{aligned}
\label{eq:authorized-dispatch}
\end{equation}
Here CAS denotes an operational event, not a pure logical query. Both exact-action and envelope certificates are single-use. The forwarded operation must be the operation validated, and retries by the gateway or transport must not silently duplicate it. Consumption and forwarding need not be atomic with each other; a crash between them can lose the operation. Nor are guard reads atomic with target execution. Section~\ref{sec:cac-theory} separates those guarantees.

\section{Formal Properties}
\label{sec:cac-theory}

The first two results explain why limited observations or resource accounting alone do not establish readiness. The remaining results state implementation obligations connecting the calculus to dispatch. They are conditional properties of the specified controller; the prototype is not a verified implementation of the entire specification.

\subsection{Information and Non-Substitution}
\label{sec:separation-theorems}

\begin{theorem}[Authority--readiness separation]
\label{thm:auth-insufficiency}
Fix $q$. Let an authorization mechanism observe only $\alpha(q,s)$. Suppose states $s_1,s_2$ share an allowed projection $\alpha_0$, but a required readiness predicate $p$ is true in $s_1$ and false in $s_2$. The mechanism allows both states and hence does not establish $p$. If CAC requires an obligation for $p$ whose satisfaction implies $p(s)$, it cannot admit $s_2$.
\end{theorem}
\begin{proof}
Equal projections give equal authorization decisions. Since $\alpha_0$ is allowed, authorization succeeds in $s_2$, where $p$ fails. Sound discharge of the required obligation would imply $p(s_2)$, a contradiction.
\end{proof}
The result concerns missing information, not the expressiveness of authorization languages. It does not apply to a mechanism whose inputs and policy already establish $p$.

\begin{theorem}[Budget does not determine admission]
\label{thm:resource-insufficiency}
Fix $q,s,\chi$ with sufficient operational budget. Suppose two evidence-acquisition procedures are feasible within the same resource envelope $R$: one returns evidence satisfying all admission conditions, and the other omits a required witness without establishing a violation. Then allocation of $R$ alone does not determine admission.
\end{theorem}
\begin{proof}
The first procedure permits admission. The second cannot satisfy the missing obligation and yields a non-permit outcome. Both are feasible under the same budget.
\end{proof}
The feasibility assumption is necessary: a zero-budget process or an unsatisfiable policy need not admit either execution. The result does not imply that additional compute cannot improve evidence acquisition.

\begin{proposition}[Compensatory scores can bypass a required condition]
\label{prop:non-fungibility}
Suppose a scalar rule admits when $f(z)\ge\theta$. A required obligation remains unsatisfied along an attainable sequence $z_n$, while $f(z_n)\to+\infty$ through unrelated proxy features. For every finite $\theta$, the scalar rule admits some $z_n$ that violates conjunctive admission.
\end{proposition}
\begin{proof}
By divergence, there is an $n$ for which $f(z_n)\ge\theta$. The required obligation remains unsatisfied by assumption, so CAC admissibility is false.
\end{proof}
This applies, for example, to a positive-weight unbounded proxy in an additive score. It excludes non-compensatory scalar encodings of conjunction and does not assume that all practical proxy features are unbounded.

\subsection{Admission and Enforcement}
\label{sec:soundness-theorems}

\begin{theorem}[Admission refinement]
\label{thm:admission-refinement}
Assume finite, schema-valid inputs; complete policy resolution; total deterministic discharge with sound witness validation; and successful guard/envelope validation. If Algorithm~\ref{alg:cac-reference} returns $\textsc{Permit}(\mathcal C_q)$, then $\operatorname{Admissible}_{\Pi}(q,s,E,\chi)$ holds and every required obligation has a corresponding valid witness in the bound manifest.
\end{theorem}
\begin{proof}
The algorithm reaches minting only after authorization and static compliance succeed. Each required obligation must either add a satisfied witness or prevent the permit path by violation, unresolved status, or failure. Complete resolution ensures no triggered obligation is omitted. The final guard/envelope check establishes the remaining admission condition. Canonical manifest construction binds exactly the collected obligation--witness pairs.
\end{proof}
This is a refinement argument for the specified algorithm, not a proof that its evidence claims are true or that a particular implementation matches every premise.

\paragraph{From admission to modeled safety.}
Let $\operatorname{Conditions}_{\Pi}(q,s)$ mean that the semantic conditions required by policy actually hold. Evidence soundness means
\[
\operatorname{Admissible}_{\Pi}(q,s,E,\chi)
\implies\operatorname{Conditions}_{\Pi}(q,s).
\]
Policy adequacy is the independent specification obligation
\begin{equation}
\begin{aligned}
\operatorname{PolicyAdequate}(\Pi,\mathcal A)\iff
\forall q,s:\;&\operatorname{Conditions}_{\Pi}(q,s)\\
&\implies\operatorname{Safe}_{\Pi}(q,s).
\end{aligned}
\label{eq:policy-adequacy}
\end{equation}
Combining the two gives modeled safety at the evaluated state. External safety additionally requires model fidelity and preservation or atomic revalidation of all decisive conditions until the effect. These assumptions cannot be derived from a certificate signature.

\begin{property}[Mediated dispatch]
\label{prop:soundness}
Assume complete mediation, authentic certificates issued only by the specified mint, correct gateway validation, and a durable linearizable nonce store shared by all gateway workers. Every governed dispatch has a valid certificate and consumes its nonce exactly once. No certificate authorizes more than one gateway dispatch.
\end{property}
\begin{proof}
Complete mediation makes the gateway the only dispatch path. Validation rejects an invalid certificate. Atomic unused-to-consumed transition has at most one successful caller per nonce; only that caller may dispatch. Unforgeability and correct minting connect an accepted certificate to an admission evaluation, except with negligible cryptographic failure probability.
\end{proof}
A crash after consumption can lose an operation. A failed or ambiguous target response does not permit reuse. Thus the property gives at-most-once gateway dispatch, not exactly-once effects, successful execution, or automatic recovery. Durability and retention must prevent nonce resurrection after restart.

\begin{property}[Execution-envelope preservation]
\label{prop:scope-preservation}
Assume collision-resistant hashing of an injective canonical encoding and correct gateway checks. In exact-action mode, a dispatched proposal matches the certified proposal. In envelope mode, it satisfies the explicit envelope predicate. If the template's semantic coverage and guard-preservation obligations hold, the manifest covers every required obligation of the dispatched variant at the evaluated snapshot.
\end{property}
\begin{proof}
A different exact-action encoding fails digest equality except with negligible collision probability. In envelope mode the gateway rejects a false envelope predicate; the universal template assumption supplies coverage and guard preservation for an accepted variant.
\end{proof}
Parameter bounds alone do not prove semantic coverage. A signature on an approved template authenticates approval, not its proof.

\begin{property}[Guard-bound dispatch validity]
\label{prop:freshness-invalidation}
If any required bound guard is false or cannot be evaluated at gateway validation time, the gateway rejects dispatch. This requires fail-closed evaluation of every bound guard and binding of the checked request to the forwarded request.
\end{property}
\begin{proof}
Guard validity is a conjunct of $\operatorname{AdmissionValid}$. Failure of a conjunct prevents the nonce-consumption and dispatch path.
\end{proof}

\paragraph{Validation versus effect time.}
Mode A checks guards before forwarding. Mode B translates relevant guards into atomic target-side preconditions, covering only the state those preconditions protect. Mode C delegates preservation through commitment to a separately justified transaction protocol. The nonce CAS serializes capability use; it does not freeze the environment or serialize changes to unrelated guard state.

\subsection{Monotonicity and Structural Cuts}
\label{sec:independence-monotonicity}

\begin{proposition}[Threshold-policy monotonicity]
\label{prop:monotonicity-threshold}
Fix proposal, policy epoch, and all non-risk inputs. If
$\Omega_{\Pi}(\rho)=\{\omega_j:\tau_j\preceq\rho\}$ for fixed obligations $\omega_j$, then
$\rho_1\preceq\rho_2$ implies
$\Omega_{\Pi}(\rho_1)\subseteq\Omega_{\Pi}(\rho_2)$.
\end{proposition}
\begin{proof}
If $\tau_j\preceq\rho_1$, transitivity gives $\tau_j\preceq\rho_2$.
\end{proof}
Replacing an obligation by a stronger predicate requires a semantic strength order, not literal set inclusion. Nothing follows for incomparable risk profiles or changing non-risk inputs.

\begin{proposition}[Structural-cut bound, imported from EFD]
\label{prop:efd-resilience}
Assume conservative exposure, closed causal accounting, authorization alignment with $\Gamma$, and that an uncorrupted verifier rejects a false proposition. If $\kappa_E(Q,\Gamma)\ge k$, fewer than $k$ active modeled roots cannot induce a decisive false approval.
\end{proposition}
\begin{proof}
Any decisive false approval requires every member of some $C\in\Gamma$ to be corrupted. Closed causal accounting and conservative exposure then require the active fault set to intersect every member's exposure set, contradicting the minimum in Equation~\eqref{eq:efd-cut-requirement}.
\end{proof}
The bound concerns modeled roots. It gives no probability of failure and no protection against omitted common causes. Appendix~\ref{app:formal-proofs} gives explicit counterexamples at the boundaries of these statements.

\section{Reference Controller and Trust Boundary}
\label{sec:reference-controller}
\label{sec:controller-arch}

The prototype separates schemas, policy resolution, evidence checks, work-loop evaluation, certificates, gateway validation, and domain adapters into TypeScript packages. The benchmark invokes them in one process. Figure~\ref{fig:reference-controller-arch} shows the intended deployment boundary; network confinement and authenticated remote interfaces are deployment requirements, not consequences of package separation.

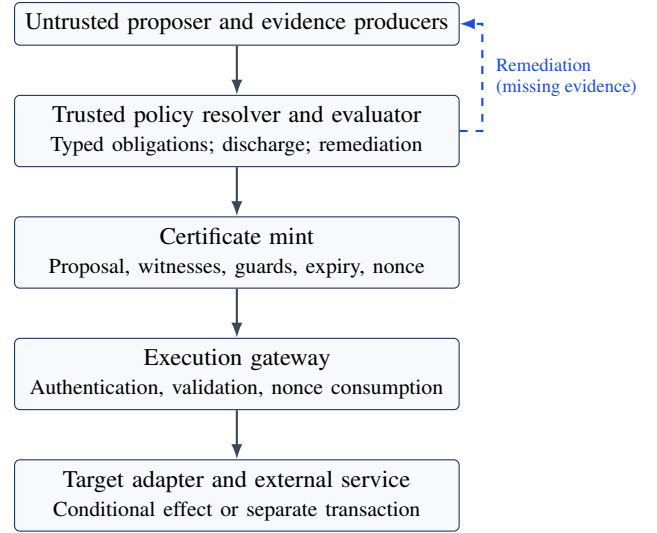
\begin{figure}[t]
\centering
\begin{tikzpicture}[
box/.style={draw=navy,rounded corners=2pt,fill=softgray,
align=center,text width=5.6cm,inner sep=4pt,font=\small},
arr/.style={-{Latex[length=2mm]},thick,darkslate},
remed/.style={-{Latex[length=2mm]},thick,dashed,cobalt}]
\node[box] (a) {Untrusted proposer and evidence producers};
\node[box,below=0.65cm of a] (b) {Trusted policy resolver and evaluator\\\footnotesize Typed obligations; discharge; remediation};
\node[box,below=0.65cm of b] (c) {Certificate mint\\\footnotesize Proposal, witnesses, guards, expiry, nonce};
\node[box,below=0.65cm of c] (d) {Execution gateway\\\footnotesize Authentication, validation, nonce consumption};
\node[box,below=0.65cm of d] (e) {Target adapter and external service\\\footnotesize Conditional effect or separate transaction};
\draw[arr] (a)--(b);
\draw[arr] (b)--(c);
\draw[arr] (c)--(d);
\draw[arr] (d)--(e);
\draw[remed] (b.east) -- ++(0.3,0) coordinate (p1)
  -- node[midway,right=1.5pt,align=left,font=\scriptsize,text=cobalt] {Remediation\\(missing evidence)} (p1 |- a.east)
  -- (a.east);
\end{tikzpicture}
\caption{Proposed deployment interface. The prototype exercises local evaluator, certificate, and gateway calls; a production boundary additionally requires credential and network confinement.}
\label{fig:reference-controller-arch}
\end{figure}

\subsection{Trust Requirements}
\label{sec:tcb-formal}

The enforcement TCB contains the policy resolver, evidence evaluator, mint, gateway, and replay store. Their correctness establishes policy-relative enforcement. External identity, time, telemetry, and topology services supply trusted inputs. Evidence-producing models remain outside the TCB, but any model output accepted as a witness inherits the policy's explicit verifier assumptions.

The gateway requires an authenticated subject binding and exclusive access to target credentials. The replay store must be shared and durable across workers. The prototype supplies both an in-memory store and a persistent local-file store. The latter uses exclusive creation followed by file and directory flushes before dispatch; an eight-process race and a restarted store produce one successful nonce consumption. This requires local-filesystem semantics and does not establish distributed linearizability, storage behavior under hardware failure, or network isolation.

\subsection{Specification Algorithm}
\label{sec:reference-algorithm}

Algorithm~\ref{alg:cac-reference} states the permit-path conditions needed by Theorem~\ref{thm:admission-refinement}. Unsupported inputs, indeterminate template coverage, and exhausted verification bounds must prevent permit.

\begin{algorithm}[t]
\caption{CAC admission specification}
\label{alg:cac-reference}
\begin{algorithmic}[1]
\Require Typed $q$, snapshot $s$, evidence $E$, context $\chi$
\If{budget exhausted or verification cannot finish}
  \State \Return $\textsc{Abort}$
\EndIf
\If{schema, authorization, or static compliance fails}
  \State \Return $\textsc{Deny}$
\EndIf
\State $\Omega\gets F_\Pi(\rho(q,s),q,s)$ \Comment{Complete resolution}
\State $\mathcal W\gets\emptyset$; $U\gets\emptyset$
\For{each required $\omega\in\Omega$}
  \State $r\gets\operatorname{Discharge}(\omega,E,s,\chi)$
  \If{$r=\textsc{Violated}(W)$}
    \State \Return $\textsc{Deny}$
  \ElsIf{$r=\textsc{Unknown}(c)$}
    \State $U\gets U\cup\{\omega\}$
  \Else
    \State Add $(\omega,W)$ from $r$ to $\mathcal W$
  \EndIf
\EndFor
\If{$U$ contains approval obligations}
  \State \Return $\textsc{Escalate}(U)$
\ElsIf{$U\ne\emptyset$}
  \State \Return $\textsc{Defer}(U)$
\EndIf
\State Derive $G_q,\mathcal E_q$ from approved policy templates
\If{coverage, guards, or issuance bounds are unverified}
  \State \Return $\textsc{Abort}$
\EndIf
\State $\mathcal C_q\gets\operatorname{Mint}(q,\mathcal W,G_q,\mathcal E_q,\chi)$
\State \Return $\textsc{Permit}(\mathcal C_q)$
\end{algorithmic}
\end{algorithm}

\subsection{Implementation Coverage}

The artifact implements typed receipts, Ed25519 signatures, policy-defined guard templates, and nonce consumption. Both positive and negative witnesses must satisfy the same set constraints. Exhaustive bounded subset search replaces singleton-or-aggregate selection, and canonical ranking uses the computed structural cut rather than source count. The direct eligibility helper and diagnostics reject future timestamps consistently. Predicate evaluation uses a bounded conjunction language over typed fields and literal comparisons; it honors declared thresholds and cannot substitute a generic success flag for missing evidence. Unsupported guards fail closed.

The controller's normal permit path mints an exact-action certificate. Envelope utilities now independently bind the action class and evidence-covered proposal, reject absent constrained values, and enforce validity bounds. They conservatively require readmission when a proposal changes, even if the new resource is allowed by the template. General semantic coverage across distinct proposals therefore remains outside the implemented fragment. Guard annotations remain authored policy, not automatically verified invariant-preservation proofs. The 133 passing tests support the exercised behaviors; they are not a proof of conformance to the entire calculus.

\section{CACBench: Controlled Local Scenarios}
\label{sec:cacbench}
\label{sec:failure-classes}
\label{sec:testbeds}

CACBench exercises the model's enforcement and availability boundaries. The corrected suite contains thirteen local failover state machines, summarized in Table~\ref{tab:cacbench-taxonomy}. It includes useful completion, harmful execution, refusal, and ambiguous-response cases. The suite is constructed around CAC's mechanisms and is not a representative sample of operational incidents.

\begin{table}[t]
\centering\small
\caption{Corrected local scenario families. Each named case has thirty parameter instances per controller.}
\label{tab:cacbench-taxonomy}
\begin{tabularx}{\columnwidth}{@{}lX@{}}
\toprule
Family & Cases and tested boundary \\
\midrule
Target state & Healthy, excessive lag, offline candidate \\
Evidence & Initially missing, unavailable, wrong class, invalid rollback \\
Dispatch change & Candidate drift, authorization revocation \\
Target response & Benign failure, committed effect with lost response \\
Source dependence & Correlated witnesses, independent witnesses \\
\bottomrule
\end{tabularx}
\end{table}

\paragraph{State and evidence.}
The world records primary and candidate identities, candidate role, unreplicated bytes, rollback validity, and effect counters. The tested policy requires zero reported replication lag, standby health, and a rollback attestation. Telemetry receipts carry source identity, class, observation time, scope, and state version. The initial-evidence and acquisition paths can be controlled independently: an underprepared but healthy target supplies missing evidence, while an unavailable source supplies none.

\paragraph{Faults at the boundary.}
Dispatch drift takes the candidate offline after initial evidence is available. Revocation removes the actor's authority at the same pre-dispatch phase. A harmful effect can still return a successful acknowledgement, forcing the evaluator to distinguish target response from ground-truth outcome. The benign-failure case returns an error before changing state; the ambiguous case changes state successfully before losing its response.

\paragraph{Correlated approval.}
Three approving sources share one modeled root in the correlated-witness case. The policy requires all three and a structural cut of at least two. The paired independent-witness case uses three disjoint roots and a healthy target. These fixtures test enforcement of declared exposure maps with synthetic receipts. They do not invoke language models, a formal solver, or a deployed telemetry service.

\paragraph{Legacy diagnostics.}
The historical F1--F11, E1, remediation, and assumption-violation fixtures remain in the artifact. Their archived logs are retained for provenance, and selected fixtures remain regression tests. They are not pooled into the corrected comparison. The new state machines are defined in \texttt{repairedScenarios.ts}; the observer and run protocol are described next.

\section{Evaluation Methodology}
\label{sec:experimental-methodology}

\subsection{Execution and Observation}
\label{sec:harness}

We evaluate the TypeScript implementation in thirteen controlled local state machines. Each scenario instantiates a failover proposal, visible state, telemetry producers, authorization state, and a target transition function. The runner retains a separate privileged world snapshot and records every attempted target effect. These experiments execute the controller and its cryptographic checks; the target is a local model, not a deployed PostgreSQL or Kubernetes service. No language-model API is invoked.

For each scenario, thirty seeds vary cluster identifiers, candidate identities, and, where applicable, the magnitude of excessive replication lag. Every controller receives a fresh world for the same scenario and seed. Controller order rotates across seeds. The matrix contains $13\times7\times30=2{,}730$ trials. Seeds instantiate parameters within constructed failure categories; they do not sample independent operational incidents.

\paragraph{Independent outcome observation.}
The runner withholds privileged world state and scoring functions from the participating controllers. Before and after each target invocation, it copies the world state and its completed-effect and harmful-effect counters. The target transition updates those counters according to the scenario's preconditions; controller verdicts and response codes do not set the outcome labels. Traces retain invocation parameters, telemetry receipts, dispatch-boundary events, pre/post snapshots, and target responses. This is separation within one trusted process, not isolation against malicious controller code.

The distinction is tested explicitly. A harmful transition can return a successful acknowledgement; a request can fail without modifying state; and an operation can complete while its response is lost. The observer records these cases separately. Ambiguous responses do not trigger an automatic retry.

\paragraph{Matched interventions.}
A shared pre-dispatch hook injects state drift or revokes authority immediately before each controller's final execution gate. CAC therefore checks live guards after the same scheduled transition that the live-policy baseline sees. Evidence availability and the telemetry API are common across controllers; the number and timing of calls are part of each strategy. This schedule tests pre-gate drift, not the later interval between forwarding and effect application.

\subsection{Comparators and Ablations}
\label{sec:baselines-formal}

\textbf{AuthOnly} checks an actual in-memory authorization snapshot and static proposal constraints before proceeding to dispatch. \textbf{LivePolicy} additionally rechecks live authorization, reacquires typed replica and rollback observations at the dispatch boundary, and evaluates the required predicates. It rejects missing evidence and affirmative counterevidence. It does not issue an admission certificate or enforce EFD structural cuts. This comparator tests the contribution relative to an environmental policy check, without claiming to reproduce a published access-control engine.

The \textbf{CAC} controller uses the evidence work loop, certificate, gateway, and declared EFD policy. Four ablations each remove one mechanism: \textbf{NoGuard} skips dispatch guards; \textbf{NoEFD} removes structural-cut requirements; \textbf{NoRemediation} denies on deferral; and \textbf{NoTypedEvidence} broadens eligible evidence classes. The typing ablation does not change tool selection. We exclude the historical extra-simulation-obligation intervention because an unavailable receipt cannot establish a general cost of static policy.

\subsection{Outcomes and Reproducibility}
\label{sec:metrics-formal}

For each intent trajectory $j$, the observer determines whether a harmful effect occurred and whether an effect completed without harm:
\begin{equation}
UIER = \frac{1}{N}\sum_{j=1}^{N}\mathbf{1}[\text{UnsafeEffect}^{*}(j)].
\label{eq:uier-metric}
\end{equation}
\begin{equation}
\begin{aligned}
SICR = \frac{1}{N}\sum_{j=1}^{N}\mathbf{1}[&\text{Complete}(j)\\
&\land\neg\text{UnsafeEffect}^{*}(j)].
\end{aligned}
\label{eq:sicr-metric}
\end{equation}
These predicates refer to the local transition model. Completion without modeled harm does not certify authorization: a failover after revocation may complete physically while violating access policy. We report counts and individual mechanism contrasts, without population-level confidence intervals or significance claims.

The run protocol fixes scenarios, seeds, controller identities, source revision, and measurement settings. The study saves full JSONL traces and a SHA-256 manifest in a new directory, refusing to overwrite existing output. An independent Python analysis checks file hashes, unique trial IDs, response categories, effect counters, aggregate counts, and timing sample counts before generating the paper's tables. The revised source passes 133 tests and strict TypeScript checking. One test compares the production structural-cut routine with direct fault propagation for all 43,561 three-voter profile/rule combinations over three roots.

\paragraph{Historical archive.}
The previous 5,261-record batch is preserved separately. Its controller-dependent labels, scripted model comparators, and incomplete timers cannot be repaired retrospectively. Its statistical and cost reports are not evidence for the results below. The corrected experiments report no model-token or monetary savings.

\section{Observed Prototype Results}
\label{sec:adversarial-evaluation}

\subsection{Completion, Refusal, and Harm}
\label{sec:hypothesis-results}

Table~\ref{tab:repaired-results} reports the corrected local matrix. CAC attempts 150 target effects across 390 trials, completes 120 without modeled harm, and produces no harmful effects. It completes the healthy, underprepared, ambiguous-response, and independent-witness scenarios in all thirty parameter instances each. Thirty additional attempts encounter the benign target failure and leave state unchanged. The remaining 240 trials do not reach the target. The zero-harm count describes these fixtures, not a bound on deployment failure probability.

\begin{table}[t]
\centering\small
\caption{Independently observed local effects across thirteen constructed scenarios and thirty parameter seeds. Completion means an observed completed effect without harm; it does not certify authorization. Counts are descriptive, not population estimates.}
\label{tab:repaired-results}
\begin{tabular}{@{}lrrrr@{}}
\toprule
Controller & Trials & Attempts & Complete & Harm \\
\midrule
AuthOnly & 390 & 390 & 180 & 180 \\
LivePolicy & 390 & 180 & 120 & 30 \\
CAC & 390 & 150 & 120 & 0 \\
NoGuard & 390 & 180 & 120 & 30 \\
NoRemediation & 390 & 120 & 90 & 0 \\
NoTypedEvidence & 390 & 180 & 120 & 30 \\
NoEFD & 390 & 180 & 120 & 30 \\
\bottomrule
\end{tabular}
\end{table}

AuthOnly attempts all 390 effects, producing 180 harmful effects and 180 completions without modeled harm. Its completion count includes physically successful operations after authorization revocation; it is not a count of policy-compliant executions. LivePolicy completes the same 120 intents as CAC, incurs the same thirty benign failures, and permits thirty harmful effects in the correlated-witness scenario. This suite therefore does not establish a general utility advantage over a competent live-policy check.

\subsection{Isolated Mechanism Contrasts}
\label{sec:ablation-results}

\paragraph{Dispatch guards.}
All controllers encounter the same candidate-offline transition before their final gate. CAC and LivePolicy prevent target invocation. NoGuard proceeds in thirty trials, each producing an observed harmful effect. This isolates the tested guard check under a matched schedule. It does not validate atomic preservation after forwarding.

\paragraph{Evidence classes.}
The wrong-class fixture supplies an authenticated static-verification receipt claiming zero lag while the world contains excessive unreplicated data. CAC refuses it as replica telemetry. NoTypedEvidence admits thirty harmful actions, with the same initial evidence and tool API. Unlike the archived F11 contrast, this intervention changes eligibility alone. It illustrates one erroneous evidence substitution, not the completeness of the evidence-class taxonomy.

\paragraph{Structural cuts.}
The correlated-witness fixture supplies three approving sources exposed to one common fault, with an all-of-three policy requiring a cut of at least two. CAC rejects the resulting cut of one. NoEFD and LivePolicy each admit thirty harmful transitions. When the three sources instead have disjoint modeled roots and report a healthy world, CAC completes all thirty trials. This checks enforcement of a supplied dependency model; it does not measure errors from actual model panels or discover hidden correlations.

\paragraph{Evidence acquisition and availability.}
CAC completes thirty initially underprepared intents after obtaining telemetry. NoRemediation completes none of those trials, reducing its total from 120 to 90. Both refuse all thirty intents for which evidence remains unavailable, even though the modeled target would otherwise be safe. Evidence availability therefore limits completion; failure to observe harm alone is insufficient to assess usefulness.

\subsection{Response Semantics and Boundaries}
\label{sec:boundary-evaluation}

In the benign-failure fixture, CAC reaches the target thirty times, receives failure responses, and records neither completion nor harm. In the ambiguous-response fixture, thirty effects complete but their responses are lost. The independent observer records completion, while the gateway returns an ambiguous outcome and does not retry. These cases show why response status, execution attempt, completion, and harm must remain separate fields.

The older assumption-violation fixtures remain diagnostic examples of omitted dependencies, false trusted evidence, and insufficient guards. The corrected study makes no claim to eliminate those boundaries. Its results establish executable enforcement and recovery behaviors in a controlled model, including matched comparisons and observed effects. Operational reliability still requires real workloads, deployment-level observers, and target-side concurrency tests.

\section{Complete Local Execution Costs}
\label{sec:performance-overhead}
\label{sec:overhead-protocol}
\label{sec:microbenchmark-design}

We measure a complete local path on macOS/ARM64 with Node.js v25.2.1 and an Apple M4 reported as ten logical CPUs. Three fresh processes each execute 1,000 measured iterations for workloads with one, four, or eight obligations and the corresponding number of distinct Ed25519-signed receipts. Twenty warm-up iterations per workload are excluded. The archive retains all 9,000 measured samples and actual host metadata. We do not assume an isolated host or pool process repetitions into a population confidence interval.

\begin{table}[t]
\centering\small
\caption{Complete local path, in milliseconds: range of per-process medians and p99s over three fresh processes, each with 1,000 measured samples per workload. Includes file-backed guard reads and persistent nonce consumption; target execution is an in-process callback.}
\label{tab:repaired-overhead}
\begin{tabular}{@{}lrr@{}}
\toprule
Obligations/receipts & Median range & p99 range \\
\midrule
1/1 & 4.65--5.69 & 10.28--37.37 \\
4/4 & 5.89--9.66 & 14.46--22.48 \\
8/8 & 24.91--26.86 & 41.41--58.39 \\
\bottomrule
\end{tabular}
\end{table}

\subsection{Measurement Boundary}
\label{sec:overhead-results}

Each iteration resolves policy, discharges every required obligation, constructs the witness manifest, and mints a fresh certificate. It then awaits a live-state JSON read from the local filesystem and invokes the gateway with that new certificate. The gateway checks signature, binding, authorization, guards, time, and policy epoch; atomically creates a spent-nonce record; flushes the file and its directory; and invokes an instrumented target callback. Every iteration asserts successful discharge, exactly one target call, successful dispatch, and a consumed nonce. A subsequent replay rejection is checked and timed separately.

The total is the elapsed interval from resolution through the awaited target return. It includes assertions within that interval, local guard-read cost, and persistent replay protection. It excludes receipt acquisition, model inference, remote service execution, and subsequent replay-check timing. The target callback increments an execution counter; it is not a database mutation. These measurements establish local implementation cost, not end-to-end deployment latency.

\subsection{Scaling and Limits}
\label{sec:cost-ttsr-discussion}

Separate sweeps use one, four, eight, and twelve distinct receipts. They time eligibility, full discharge, canonical encoding, and a manifest containing the entire receipt set; no receipt count is silently capped. Full-discharge timing includes eligibility and canonical selection as well as subset search, so these components must not be interpreted as a non-overlapping decomposition. A separate EFD sweep varies both voters and roots over two, four, six, and eight, with disjoint exposures and a two-of-$n$ rule. The computed cut is two at every size. Each process repeats each scaling point five times.

The discharge implementation enumerates nonempty subsets and checks both polarities. It is deliberately bounded at twelve unique eligible receipts per obligation; larger pools return unknown. Exact cut computation is also combinatorial. These measurements do not establish high-throughput behavior for large coalitions. Stronger witness search and durable replay protection cost more than an unchecked primitive call, and remote evidence acquisition may dominate both. A deployment evaluation must measure that additional work and effect-time concurrency under its actual storage and service configuration.

The full-discharge sweep illustrates the cost of that bound: per-process medians increase from 0.072--0.079\,ms for one receipt to 67.31--72.69\,ms for twelve. This supports explicit admission input limits, not an extrapolation to thousands of receipts. Variation across the three full-path processes also cautions against treating a single median or p99 as a stable service-level objective.

\section{Integration with Execution Systems}
\label{sec:architectural-integration}
\label{sec:rel-tct}

Admission and commitment address different stages of an operation within post-deterministic distributed systems (\pdds)~\cite{he2026pdds}. A migration can have a conflict-free read set yet lack a policy-required backup receipt. Conversely, an admitted migration can conflict with another transaction before commitment. A CAC certificate therefore supplies authorization to attempt the bounded operation; the storage or workflow runtime must still enforce its own isolation and settlement rules~\cite{kung1981occ,ports2012ssi,atomix2026}.

\paragraph{Target-side preconditions.}
A target that supports atomic conditional mutation can check certified versions at the same linearization point as the effect. This covers only the predicates represented by those preconditions. Checking one resource version does not atomically validate an arbitrary multi-resource guard set. When guards span services, their preservation requires a suitable transaction or coordination protocol; generic composition alone provides no serializability theorem.

\paragraph{Capability gateways.}
CAC can supply the evidence-discharge decision to a capability broker such as SAB~\cite{he2026sab}. The broker supplies authentication, confinement, and revocation; the admission policy supplies the required evidence and its interpretation. Such integration remains an architectural interface in this work, not an evaluated deployment.

\paragraph{Policy updates.}
Learned lessons and incident analyses may propose new obligations. They cannot alter the active policy without the governing approval process. Policy epochs, evidence-source registrations, and envelope templates are trusted configuration. A revision to any of them may invalidate previously issued certificates.

\section{Discussion and Limitations}
\label{sec:discussion-limitations}
\label{sec:discussion-theory}
\label{sec:limitations-formal}

\paragraph{Policy adequacy.}
CAC exposes the obligations used to justify an action; it does not infer a complete safety specification from natural-language intent. Policy authors must connect each obligation to a domain invariant and justify the associated guards. An omitted dependency, permissive entailment rule, or incomplete envelope proof can preserve syntactic admission correctness while permitting harm.

\paragraph{Evidence trust.}
Signatures establish provenance under a key-trust model, not the truth of a sensor reading. Structural cuts provide a lower bound on modeled fault coverage only when exposure maps conservatively include the causes of erroneous approval and the coalition rule matches actual authorization. Undisclosed common dependencies can invalidate that bound. Reconciliation can detect some discrepancies but cannot repair arbitrary coordinated falsehoods.

\paragraph{Availability and cost.}
A conjunctive policy can block indefinitely if a required witness is unavailable. Operational budgets bound attempts but do not guarantee completion. Recursive admission for consequential diagnostics needs a well-founded dependency discipline or a bounded failure path. Evidence collection can dominate controller processing, and false refusal matters as much as unsafe dispatch when evaluating an operational system.

\paragraph{Prototype scope.}
The corrected study uses independently observed local state transitions, a live-policy comparator, and matched interventions. Its target models still encode a small set of known preconditions and serialize each effect. They do not implement database replication, fencing, network partitions, or service recovery. The complete-path timings include real local file operations but an in-process target callback. These limits preclude claims of production reliability, model-family generalization, economic savings, or high-throughput EFD scaling.

\paragraph{Validation needed.}
A deployment evaluation should extend the implemented observer and matched-schedule protocol to separate services. It should include realistic safe workloads, expensive sensors, harmful partial failures, process crashes, distributed replay stores, and post-forwarding races. The present multi-process replay test and performance repetitions address local behavior; representative scenario sampling is still required for population inference. Comparisons should also include a neighboring deterministic admission implementation and real proposer/verifier execution. General envelope coverage needs a domain-specific proof or a checked coverage relation before relaxed action binding is enabled.

\paragraph{Long-running workflows.}
A single admission certifies one bounded dispatch. Multi-step migrations or progressive rollouts need fresh admission at relevant transitions and separate recovery semantics. Continuous monitoring may complement that design, but its safety and liveness properties are outside the present model.

\section{Related Work}
\label{sec:related-work}

\paragraph{Authorization and proof-carrying execution.}
Complete mediation and separation of privilege are established protection principles~\cite{saltzer1975protection}. ABAC explicitly includes environmental conditions~\cite{hu2014abac}; it is not restricted to static identity checks. CAC's predicates can be implemented in a sufficiently expressive policy engine. The proposed distinction is the contract for acquiring, retaining, and revalidating evidence for those predicates.

Proof-carrying code requires an untrusted producer to provide a consumer-checkable proof of compliance with a safety policy~\cite{necula1997pcc}. CAC shares that producer/checker separation. Its witnesses also include time-sensitive observations and policy-governed verifier receipts, which need not be formal proofs of a program's semantics. CAC adds a missing-evidence protocol and dispatch binding; it inherits rather than replaces the need to justify the consumer's policy.

\paragraph{Agent admission and repair.}
Mnemosyne separates generated proposals from deterministic admission under executable constraints and supports bounded repair, compensation, and active contract records~\cite{mnemosyne2026}. These overlap directly with CAC's enforcement boundary and work loop. We therefore do not claim deterministic admission, obligations, or repair in isolation as novel. CAC concentrates on risk-conditioned evidence selection, receipt-level freshness and provenance, and structural witness requirements. Whether this interface improves upon a configured Mnemosyne-style runtime requires an implementation comparison absent from the current study.

\paragraph{Commitment and freshness.}
Commit-time authorization binds durable effects to authority evidence that remains eligible at commitment~\cite{commit_authorization2026}. CAC shares freshness and effect-binding requirements but specifies which evidence must first be acquired. Pre-dispatch checking is weaker than commitment-time checking unless the target preserves the bound conditions. Atomix coordinates effect grouping and progress-aware settlement~\cite{atomix2026}; database concurrency control supplies separate isolation mechanisms~\cite{kung1981occ,ports2012ssi}. Admission does not subsume either form of settlement.

\paragraph{Verifier structure and capabilities.}
The structural-cut definition and fault-coverage bound are imported from EFD~\cite{he2026efd}, whose controller already enforces cuts for quorum admission. CAC integrates that requirement with other typed obligations; it does not introduce EFD-based admission itself. Capability boundaries such as SAB~\cite{he2026sab} address the broker and credential interface. Their enforcement properties remain independent deployment assumptions here.

\paragraph{Reasoning and verification.}
ReAct interleaves reasoning and actions, Reflexion uses feedback, and process supervision and test-time compute can improve reasoning quality~\cite{yao2022react,shinn2023reflexion,lightman2023letsverify,snell2024scaling}. CAC can consume evidence produced by such methods. Its logical separation result says that resource allocation alone does not certify an external condition; it does not show that these methods are ineffective, nor does the scripted B0--B5 comparison measure them.

\section{Conclusion}
\label{sec:conclusion}

Cognitive Admission Control makes evidence requirements explicit at the execution boundary between autonomous agent proposals and external services in agentic distributed systems. A policy maps action risk to typed obligations; the controller distinguishes satisfied, violated, and unresolved conditions, requests missing evidence, and binds a successful evaluation to a scoped certificate and dispatch guards. Its guarantees concern admission and mediated dispatch under stated assumptions. Physical safety additionally depends on the adequacy of the policy, evidence sources, environment model, and target-side enforcement.

The repaired TypeScript prototype exercises bounded witness search, independently bound certificates, live guards, and persistent local replay protection. Controlled experiments observe effects independently, isolate selected mechanisms, and measure the complete local dispatch path. CAC and a live-policy baseline achieve equal completion on the tested suite, while structural-cut enforcement blocks the constructed correlated-witness failure. CAC's practical value now depends on whether deployed systems can obtain the required evidence at acceptable cost while completing useful work under realistic faults.

\smallskip
\noindent\textbf{AI-Use Disclosure.} OpenAI Codex and Google Antigravity assisted with LaTeX formatting, draft structuring, language editing, implementation and test development, notation and schema consistency checks, analysis scripting, and figure preparation. The authors remain responsible for the study, code, and reported results. All reported measurements come from executed code and frozen artifacts rather than model-generated estimates.

\bibliographystyle{unsrt}
\bibliography{refs}

@techreport{hu2014abac,
  author = {Hu, Vincent C. and Ferraiolo, David and Kuhn, Rick and Schnitzer, Adam and Sandlin, Kenneth and Miller, Robert and Scarfone, Karen},
  title = {Guide to Attribute Based Access Control ({ABAC}) Definition and Considerations},
  institution = {National Institute of Standards and Technology},
  number = {SP 800-162},
  year = {2014},
  note = {Updated August 2019},
  doi = {10.6028/NIST.SP.800-162}
}

@inproceedings{necula1997pcc,
  author = {Necula, George C.},
  title = {Proof-Carrying Code},
  booktitle = {Proceedings of the 24th ACM SIGPLAN-SIGACT Symposium on Principles of Programming Languages},
  pages = {106--119},
  year = {1997},
  doi = {10.1145/263699.263712}
}

@article{saltzer1975protection,
  author  = {Saltzer, Jerome H. and Schroeder, Michael D.},
  title   = {The Protection of Information in Computer Systems},
  journal = {Proceedings of the IEEE},
  volume  = {63},
  number  = {9},
  pages   = {1278--1308},
  year    = {1975},
  doi     = {10.1109/PROC.1975.9939}
}

@article{yao2022react,
  author  = {Yao, Shunyu and Zhao, Jeffrey and Yu, Dian and Du, Nan and Shafran, Izhak and Narasimhan, Karthik and Cao, Yuan},
  title   = {{ReAct}: Synergizing Reasoning and Acting in Language Models},
  journal = {arXiv preprint arXiv:2210.03629},
  year    = {2022}
}

@article{shinn2023reflexion,
  author  = {Shinn, Noah and Cassano, Federico and Gopinath, Ashwin and Narasimhan, Karthik and Yao, Shunyu},
  title   = {Reflexion: Language Agents with Verbal Reinforcement Learning},
  journal = {Advances in Neural Information Processing Systems},
  volume  = {36},
  pages   = {8634--8652},
  year    = {2023}
}

@article{snell2024scaling,
  author  = {Snell, Charlie and Lee, Jaehoon and Xu, Kelvin and Kumar, Aviral},
  title   = {Scaling {LLM} Test-Time Compute Optimally Can Be More Effective than Scaling Model Parameters},
  journal = {arXiv preprint arXiv:2408.03314},
  year    = {2024}
}

@article{lightman2023letsverify,
  author  = {Lightman, Hunter and Kosaraju, Vineet and Burda, Yura and Edwards, Harri and Baker, Bowen and Lee, Teddy and Leike, Jan and Schulman, John and Sutskever, Ilya and Cobbe, Karl},
  title   = {Let's Verify Step by Step},
  journal = {arXiv preprint arXiv:2305.20050},
  year    = {2023}
}

@article{wei2022chain,
  author  = {Wei, Jason and Wang, Xuezhi and Schuurmans, Dale and Bosma, Maarten and Xia, Fei and Chi, Ed and Le, Quoc V. and Zhou, Denny},
  title   = {Chain-of-Thought Prompting Elicits Reasoning in Large Language Models},
  journal = {Advances in Neural Information Processing Systems},
  volume  = {35},
  pages   = {24824--24837},
  year    = {2022}
}

@article{commit_authorization2026,
  author  = {Santos-Grueiro, Igor},
  title   = {Temporary Authority, Permanent Effects: Commit-Time Authorization for {LLM} Agents},
  journal = {arXiv preprint arXiv:2607.10487},
  year    = {2026}
}

@article{atomix2026,
  author  = {Mohammadi, Bardia and Potamitis, Nearchos and Klein, Lars and Arora, Akhil and Bindschaedler, Laurent},
  title   = {Atomix: Timely, Transactional Tool Use for Reliable Agentic Workflows},
  journal = {arXiv preprint arXiv:2602.14849},
  year    = {2026}
}

@article{mnemosyne2026,
  author  = {Chang, Edward Y. and Geng, Longling and Chang, Emily J.},
  title   = {Mnemosyne: Agentic Transaction Processing for Validating and Repairing {AI}-generated Workflows},
  journal = {arXiv preprint arXiv:2607.00269},
  year    = {2026}
}

@article{he2026pdds,
  author  = {He, Jun and Yu, Deying},
  title   = {Post-Deterministic Distributed Systems: A New Foundation for Trustworthy Autonomous Infrastructure},
  journal = {arXiv preprint arXiv:2606.01722},
  year    = {2026}
}

@article{he2026efd,
  author  = {He, Jun and Yu, Deying},
  title   = {The Illusion of Independent Quorums: Epistemic Fault Domains and Correlated Cognitive Failures in Agentic Quorums},
  journal = {arXiv preprint arXiv:2609.02925},
  year    = {2026}
}

@article{he2026sab,
  author  = {He, Jun and Yu, Deying},
  title   = {Sovereign Assurance Boundary: Certificate-Bound Admission for Agentic Infrastructure},
  journal = {arXiv preprint arXiv:2606.11632},
  year    = {2026}
}

@article{ports2012ssi,
  author  = {Ports, Dan R. K. and Grittner, Kevin},
  title   = {Serializable Snapshot Isolation in {PostgreSQL}},
  journal = {Proceedings of the VLDB Endowment},
  volume  = {5},
  number  = {12},
  pages   = {1850--1861},
  year    = {2012},
  doi     = {10.14778/2367502.2367523}
}

@article{kung1981occ,
  author  = {Kung, H. T. and Robinson, John T.},
  title   = {On Optimistic Methods for Concurrency Control},
  journal = {ACM Transactions on Database Systems},
  volume  = {6},
  number  = {2},
  pages   = {213--226},
  year    = {1981},
  doi     = {10.1145/319566.319567}
}

\appendix
\section{Formal Boundary Examples}
\label{app:formal-proofs}

\subsection{Projection and Budget Counterexamples}
\label{app:proof-authority-readiness}
\label{app:proof-scalarization}

For Theorem~\ref{thm:auth-insufficiency}, consider two snapshots with the same principal, database identifier, and failover permission. The candidate is fenced in one snapshot and unfenced in the other. Authorization based on those shared fields cannot distinguish them. The conclusion requires a sound fencing obligation; a falsely signed assertion can defeat the safety premise while preserving syntactic admission.

For Theorem~\ref{thm:resource-insufficiency}, let both procedures have access to one observation call. One calls a source that establishes the sole required predicate; the other reads unrelated documentation. Equal allocation permits different evidence outcomes. If the required source is inaccessible to both, the example does not establish that either can obtain admission.

For Proposition~\ref{prop:non-fungibility}, take $f(x,y)=x+y$, threshold $1$, and require $x\ge1$. Along $(0,n)$ the required condition fails while the score eventually exceeds the threshold. In contrast, the scalar indicator $\mathbf1[x\ge1\land y\ge1]$ faithfully represents conjunction. The issue is substitutability, not scalar notation itself.

\subsection{Monotonicity Scope}
\label{app:proof-monotonicity}

The threshold proposition fixes the rule base and all non-risk inputs. A policy that removes a backup obligation when a new emergency flag becomes true is not covered even if its risk estimate increases. Similarly, replacing a five-second freshness obligation by a one-second obligation strengthens policy semantically but does not preserve literal identity of obligation records.

\subsection{Quorum Extension}
\label{app:proof-efd}

Let $Q=\{a,b,c\}$ use a two-of-three rule with $\mathcal D(a)=\mathcal D(b)=\{f\}$ and $\mathcal D(c)=\{g\}$. Coalition $\{a,b\}$ is decisive and covered by $\{f\}$, so $\kappa_E=1$. Add $d$ with exposure $\{h\}$ while keeping a two-vote threshold. Coalition $\{a,b\}$ remains decisive, hence the cut remains one. An independent new source cannot repair a coalition rule that still permits approval by the original correlated pair.

By contrast, requiring both members of $\{c,d\}$ gives a cut of two when their exposure sets are disjoint and the declared basis covers all error causes. This is a property of the coalition structure and exposure model together.

\subsection{Conflicts, Expiration, and Envelopes}

If equally preferred valid witnesses support opposite polarities, deterministic tie-breaking must not select a polarity by receipt identifier. The verdict is unknown; identifier ordering may select a canonical witness only after the semantic decision.

A receipt dated after the evaluation instant satisfies an upper-age bound alone, even if arbitrarily far in the future. The lower bound in Equation~\eqref{eq:eligibility-def} excludes it. Cross-host time comparison requires the declared normalization assumptions.

Finally, an exact proposal digest prevents substitution but does not derive a missing guard. If a policy omits primary fencing, hashing the exact failover request does not make the omission safe. Envelope soundness requires witness coverage and guard completeness in addition to syntactic action binding.

\subsection{Bounded Exhaustive Cross-Checks}

The independent script \path{scripts/check_finite_model.py} enumerates three voters and three modeled roots, all 343 nonempty exposure profiles, and all 127 nonempty families of nonempty coalitions. Across 43,561 profile/rule combinations and 348,488 fault states, minimum coalition coverage agrees with direct propagation of faults to voters, and no sub-cut fault set covers a decisive coalition. The script also verifies the quorum-extension counterexample and threshold monotonicity for all 216 comparable pairs in a three-coordinate, three-level risk domain. These finite checks supplement the arguments above; they neither prove the unbounded results nor test the TypeScript implementation.

\section{Illustrative Policy Contract}
\label{app:obligation-catalogue}
\label{app:schema}

A failover policy should distinguish the observed lag condition from the fencing condition. The following design example uses fields corresponding to the prototype's obligation records. It is an excerpt, not a complete deployable profile; source registration, predicate interpretation, version compatibility, and guard bindings must also be supplied.

\begin{lstlisting}[float=t,caption={Illustrative replication observation obligation; the two-second horizon is a policy example, not an evaluated optimum.},label={lst:obligation-schema}]
{
  "id": "replication-current",
  "kind": "OBSERVE",
  "predicate": "replication_lag_bytes <= 0",
  "target": "cluster.candidate",
  "scope": ["postgres/prod-cluster-a"],
  "maxFreshnessMs": 2000,
  "evidenceClasses": ["POSTGRES_TELEMETRY"],
  "efdRequirement": null,
  "setConstraints": null,
  "enforcement": "REQUIRED"
}
\end{lstlisting}

The claim must identify the primary and candidate and the replication positions it compares. A role-only observation cannot establish this predicate. Policy separately requires evidence that the old primary is fenced and that the chosen target remains the intended standby. Guards must preserve the relevant conditions or translate them into target-side atomic preconditions. Zero observed lag at one instant does not by itself exclude a concurrent write before fencing.

A simple trace illustrates the three discharge outcomes. With no receipt, the replication obligation is unknown and the controller requests telemetry. An eligible receipt affirmatively showing positive lag violates this example's zero-lag requirement. An eligible zero-lag receipt satisfies this obligation, but admission still waits for fencing and every other required condition. A valid certificate becomes unusable when a bound version changes or the earliest supporting receipt expires.

The archived R1 fixture uses a positive lag tolerance and a fixed safe oracle. The corrected local study evaluates a literal zero-lag predicate and observes target effects separately, but its serial transition model does not validate primary fencing or exclude concurrent writes in a real database.

\section{Open Problems}
\label{app:unresolved-questions}

\paragraph{Guard sufficiency.}
Automatically deriving a minimal guard set requires a semantic connection from each witness to the state on which its claim depends. A dependency annotation supplies a candidate, not a proof that all relevant changes are covered. Multi-resource guards also need a target-side preservation mechanism.

\paragraph{Conservative risk classification.}
The policy must account for indirect effects and incomplete dependency graphs. A typed operation narrows the interpretation problem but does not determine all consequences. Unknown scope or unresolved aliases should produce a conservative bound or escalation, not an unqualified low-risk classification.

\paragraph{Bounded evidence search.}
General witness selection and fault-cut computation can be combinatorial. An implementation needs a restricted evidence language and explicit resource bounds. A useful research target is a sound conservative approximation that improves availability without treating incomplete search as proof of satisfaction.

\paragraph{Remediation liveness.}
A diagnostic action can depend on admission of another diagnostic action. Shared budgets bound attempts, but deadlock freedom and successful completion require additional assumptions about the dependency graph, sensor availability, and environmental stability.

\paragraph{Fault-profile validation.}
Exposure maps must match actual evidence paths and authorization coalitions. Attestation may establish a declared configuration, while empirical correlation studies can expose some missing dependencies. Neither method by itself proves that all common causes have been enumerated.

\end{document}